\documentclass{article}
\pdfoutput=1

\usepackage{arxiv}

\usepackage[utf8]{inputenc} 
\usepackage[T1]{fontenc}    
\usepackage{hyperref}       
\usepackage{url}            
\usepackage{booktabs}       
\usepackage{amsfonts}       
\usepackage{nicefrac}       
\usepackage{microtype}      
\usepackage{mathtools}
\usepackage{amsthm}
\usepackage{graphicx}

\newtheorem{prop}{Proposition}[section]
\newtheorem{rem}{Remark}[section]
\newtheorem{exmp}{Example}[section]

\numberwithin{equation}{section}
\title{{Necessary and Sufficient Conditions for Time Reversal Symmetry in Presence of Magnetic Fields}}

\author{
  Davide Carbone\\
  Dipartimento di Fisica\\ Università di Torino\\
  Via Pietro Giuria 1, 10125 Torino, Italy\\ \\
  email: \texttt{davide.carbone@edu.unito.it}
   \And
 Lamberto Rondoni \\
  Dipartimento di Scienze Matematiche\\ Politecnico di Torino\\ Corso Duca degli Abruzzi 24, 10129 Torino, Italy \\ \\Istituto Nazionale di Fisica Nucleare\\ Sezione di Torino\\Via Pietro Giuria 1, 10125 Torino, 
 Italy \\ \\
 email: \texttt{lamberto.rondoni@polito.it} \\
}

\begin{document}
\maketitle

\begin{abstract}
Time reversal invariance (TRI) of particles systems has many consequences, among~which the celebrated Onsager reciprocal relations, a milestone in Statistical Mechanics dating back to 1931. Because for a long time it was believed that (TRI) dos not hold in presence of a magnetic field, a modification of such relations was proposed by Casimir in 1945. Only in the last decade, the~strict traditional notion of reversibility that led to Casimir's work has been questioned. It was then found that other symmetries can be used, which allow the Onsager reciprocal relations to hold without modification. In this paper we advance this investigation for classical Hamiltonian systems, substantially increasing the number of symmetries that yield TRI in presence of a magnetic field. We~first deduce the most general form of a generalized time reversal operation on the phase space of such a system; secondly, we express sufficient conditions on the magnetic field which ensure TRI. Finally, we examine common examples from statistical mechanics and molecular dynamics. Our main result is that TRI holds in a much wider generality than previously believed, partially explaining why no experimental violation of Onsager relations has so far been reported.
\end{abstract}

\keywords{hamiltonian dynamics \and magnetic field \and correlation functions \and onsager reciprocal relations}

\section{Introduction}
\label{sec:Introduction}
The relation between time reversal invariance (TRI) and Onsager reciprocal relations \cite{onsager1931reciprocal,onsager1931reciprocal2}, for~systems coupled with a magnetic field is a topic well investigated since Casimir's article \cite{RevModPhys.17.343}. A cardinal contribution was given by Kubo in Refs. \cite{kubo1957statistical,kubo1959some,kubo1966fluctuation} who used the usual time reversal operation
\begin{equation}
\label{eq:0}
   {\mathcal{T}_B(\boldsymbol{r},\boldsymbol{p},t;\boldsymbol{B})=(\boldsymbol{r},-\boldsymbol{p},-t;-\boldsymbol{B})}
\end{equation} for the correlator of two classical observables $\phi$ and $\psi$ in the stationary state, where $\boldsymbol{r},\boldsymbol{p}$ collectively represent coordinates and momenta of the particles of the system of interest. He obtained the following chain of equalities:
\begin{equation}
    \label{eq:1}
    \langle\phi(0)\psi(t)\rangle_{\boldsymbol{B}}=\eta_\phi\eta_\psi\langle\phi(0)\psi(-t)\rangle_{-\boldsymbol{B}}=\eta_\phi\eta_\psi\langle\phi(t)\psi(0)\rangle_{-\boldsymbol{B}}
\end{equation}

Here the factors $\eta_\psi$ and $\eta_\phi$ are, respectively, the signatures of the {\em observables}  {$\psi$ and $\phi$}, 
 i.e., of two generic functions defined on the phase space, with regard to the transformation $\mathcal{T}_B$. Moreover, the~angular brackets represent the average with respect to the equilibrium probability distribution in phase space.

Generalized time reversal transformations different from $\mathcal{T}_B$ are already given by Lax in Ref. \cite{lax}, but in the previous century the statement that crystallized in the literature was that only $\mathcal{T}_B$ allows the reciprocal relations to hold. Unfortunately, this only leads to a relation between two different systems as stressed by the subscripts in \eqref{eq:1}, one with magnetic field $\boldsymbol{B}$ and the other with opposite field, which leads to Casimir's modification of Onsager reciprocal relations. As a consequence, the predictive power of these relations is quite limited, compared to that of the original relations.

Recently, however, a different perspective has been adopted in Refs. \cite{rondoe,rondom,bonella2017time} for classical systems coupled with a constant magnetic field along an axis and in Ref. \cite{casa} for a magnetic field dependent on one space coordinate. In particular, it was shown that suitable time reversal operations exist that yield~\eqref{eq:1} without the inversion of the field. Furthermore, the quantum case, in the presence of a constant magnetic field has been similarly treated in Ref. \cite{rondoquantu}.

As we will show in detail, the generalized time reversal transformations that were investigated do not exhaust the set of all possible operations leading to TRI. The first objective of this paper is to identify the most general time reversal operation compatible with a classical Hamiltonian system. After~this, we analyze the minimal coupling with a generic magnetic field, formulating sufficient conditions for the magnetic field and for the force potential that make the Onsager reciprocal relations~hold. 

This theoretical result is relevant also in the context of quantum mechanics, that will be dealt with in a future paper. For exemplary instance, in Ref. \cite{jacquod2012onsager} B{\"u}ttiker and collaborators analyzed quantum systems using the ``tenfold way'' developed by Zirnbauer in Ref. \cite{zirnbauer2010symmetry}, which is founded on the idea that the validity of the Onsager reciprocal relations necessarily requires microreversibility,  i.e.,  Onsager's notion that: ``if~the velocities of all the particles present are reversed simultaneously the particles will retrace their former paths, reversing the entire succession of configurations'', which is to say that $\mathcal{T}(\boldsymbol{r},\boldsymbol{p},t)=(\boldsymbol{r},-\boldsymbol{p},-t)$ holds. As demonstrated in Refs. \cite{rondoe,rondom,bonella2017time}, this is not always required for statistical properties, because other symmetries may as well do. In this paper we show that further generalized time reversal operations exist that can be used in Linear Response Theory and beyond.

In Section \ref{sec:Theory and results}, we derive and discuss our results about time reversal invariant (TRI) systems, in~presence of magnetic fields, and we introduce our methods of investigation. In particular, we~provide sufficient conditions for the magnetic fields that allow TRI. In Section \ref{sec:Conclusions}, we summarize such results and outline future developments.
 
\section{Theory and Results}
\label{sec:Theory and results}
This section is organized as follows: Section \ref{sec:Onsager reciprocal relations and T-symmetry} summarizes previous results on TRI in presence of a magnetic field and its relevance for the Onsager reciprocal relations and other statistical equalities. Section \ref{sec:Dynamics and Transformations} identifies the general form of a TRI operation for a system coupled with a magnetic field $\bf B$, and gives sufficient conditions on $\bf B$ for such operations to exist. This is connected with the question of gauge freedom, which is analyzed in Section \ref{sec:Gauge}. Section \ref{sec:Magnetic Field} closes the loop concerning sufficient conditions, expressing them directly from the point of view of the magnetic field. Finally, various~examples of potentials are used to illustrate our theoretical results.
\subsection{Onsager Reciprocal Relations and T-Symmetry}
\label{sec:Onsager reciprocal relations and T-symmetry}
\noindent

A dynamical system $S^t : \Omega \to \Omega$, on a phase space $\Omega$ with $t \in \mathbb{R}$, is called TRI if there exists a map $\mathcal{M}:\Omega\xrightarrow{}\Omega$, {such that:} 
\begin{equation}
    \label{timerev1}
    \mathcal{M}S^t=S^{-t}\mathcal{M} \, , \quad \mbox{and }  ~~~  \mathcal{M}^2=I
\end{equation}

The operator $S^t$ is the time evolution operator on the phase space, which moves every initial condition $\Gamma \in \Omega$ to the corresponding evolved phase point $S^t \Gamma \in \Omega$. As $S^t$ and $S^{-t}$ are operators related to the same dynamics, forward in one case and backward in the other, $\mathcal{M}$ in \eqref{timerev1} has to preserve the equations of motion and so the Hamiltonian, cf.\ Section~\ref{sec:Dynamics and Transformations}. 

{As shown for instance in Ref. \cite{rondoe}, the canonical time reversal operation,  i.e., $\mathcal{M}(\boldsymbol{r},\boldsymbol{p})=(\boldsymbol{r},-\boldsymbol{p})$, does not verify Equation \eqref{timerev1} when $S^t$ describes the evolution of a system in a magnetic field. While~the equations of motion are preserved by $\mathcal{T}_B$,  i.e., by inverting momenta and magnetic field together with time, that operation}
 means dealing with different systems, subject to different magnetic fields, rather~than with a single system in given magnetic field. Thus,  {one only obtains relations such as the Onsager--Casimir ones, \eqref{eq:1}, that do not quantify the properties of a system of interest: they merely link non-quantified properties of two different systems in two different magnetic fields.}

Given the observables  
$\phi, \psi :\Omega\xrightarrow{}\mathbb{R}$, their correlator with respect to a probability distribution in phase space, $\rho$, is defined by:
\begin{equation}
    \label{correl:1}
    \langle\phi(0)\psi(t)\rangle_{\boldsymbol{B}}=\int_\Omega dX\rho(X)\phi(X)\psi(S^t X)
\end{equation}

In case an operation $\mathcal{M}$ verifying Equation \eqref{timerev1} exists, Onsager reciprocal relations hold, as can be demonstrated analyzing the correlator \eqref{correl:1}. This can be seen through the following steps: first, $\mathcal{M}$~is used to change variable within the integral, setting $X=\mathcal{M}Y$, whose Jacobian determinant is $1$, because $\cal M$ is an isometry. It follows that:
\begin{equation}
    \label{correl:2}
    \langle\phi(0)\psi(t)\rangle_{\boldsymbol{B}}=\int_\Omega dY\rho(\mathcal{M}Y)\phi(\mathcal{M}Y)\psi(S^t \mathcal{M}Y)
\end{equation}

Suppose that $\phi$ and $\psi$  respectively possess signatures $\eta_\phi$ and $\eta_\psi$ under the action of $\mathcal{M}$, and~that the probability density $\rho$ is even under $\mathcal{M}$, as appropriate for an equilibrium distribution of a Hamiltonian particles system, such as the canonical ensemble. This leads to the result showed in~Ref.~\cite{rondoe}:
\begin{equation}
    \label{correl:3}
    \langle\phi(0)\psi(t)\rangle_{\boldsymbol{B}}=\eta_\phi\eta_\psi\int_\Omega dY\rho(Y)\phi(Y)\psi(S^{-t} Y)=\eta_\phi\eta_\psi\langle\phi(0)\psi(-t)\rangle_{\boldsymbol{B}}
\end{equation}

Using the invariance for time translation of the equilibrium state,  i.e., translating forward by a time $t$ the last term of \eqref{correl:3}, we come to the final result:
\begin{equation}
    \label{correl:4}
    \langle\phi(0)\psi(t)\rangle_{\boldsymbol{B}}=\eta_\phi\eta_\psi\langle\phi(t)\psi(0)\rangle_{\boldsymbol{B}}
\end{equation}

This is related to the Onsager theory of linear response as follows: given the  macroscopic observables $\alpha_i$, $i=1,...,n$, and entropy $\mathcal{S}$ of a system subjected to (relatively) small thermodynamic forces $X_j$, $j=1,...,n$, one may write:
\begin{equation}
    \label{ons:1}
    \dot{\alpha_i}=\sum_j L_{ij}X_j \;\;\;\;\;\;\; X_j=\frac{\partial\mathcal{S}}{\partial\alpha_j}
    \, ; \quad i,j = 1, ... , n
\end{equation}
where the linear transport coefficients are obtained via the Green--Kubo integrals of the
corresponding correlators (see Ref. \cite{Bettolo}). Therefore, the symmetry properties of $L_{ij}$ descend from those of $\langle\alpha_i(0)\alpha_j(t)\rangle$. If $\eta_i$ and $\eta_j$ are the signatures of the macroscopic observables, we have: 
\begin{equation}
    \label{ons:2}
    \langle\alpha_i(0)\alpha_j(t)\rangle_{\boldsymbol{B}}=\eta_i\eta_j\langle\alpha_i(t)\alpha_j(0)\rangle_{\boldsymbol{B}} 
    \, ; \quad i,j = 1, ... , n
\end{equation}
that, after integration in time, yield the Onsager reciprocal relations:
\begin{equation}
L_{ij}=\eta_i\eta_jL_{ji} 
\, ; \quad i,j = 1, ... , n
\end{equation}

Our goal is to identify the general form of a time reversal transformation, as well as the conditions under which Onsager symmetry may be obtained in presence of a magnetic field.

\subsection{Dynamics and Transformations}
\label{sec:Dynamics and Transformations} 
Consider a system of particles coupled with an external static magnetic field and subject to forces expressed by a potential. The corresponding Hamiltonian writes:
\begin{equation}
\label{eq:2.2}
   H=\sum_{i=1}^N\left[\frac{(\boldsymbol{p}_i-q_i\boldsymbol{A}(x_i,y_i,z_i))^2}{2m_i}\right]+U(\boldsymbol{X},\boldsymbol{P},\boldsymbol{C}) 
\end{equation}
where $N$ is the number of particles, $q_i$ and $m_i$ are the charge and the mass of the $i$-th particle, the~first addend is the coupling to the magnetic field and $U(\boldsymbol{X},\boldsymbol{P},\boldsymbol{C})$ is the force potential. In general, $U$~depends on $2dN$ coordinates $(\boldsymbol{X},\boldsymbol{P})$, if each particle has got $d$ degrees of freedom, but it may also depend on a set of parameters $\boldsymbol{C}$. Without loss of generality, let us assume that the particles move in 3-dimensional space and that $d=3$. In the following we are going to use $A_k(x_i,y_i,z_i)$, with $k=1,2,3$, to denote the components of the vector potential $\boldsymbol{A}(x_i,y_i,z_i)$.

Let us begin identifying the possible time reversal operations for a Hamiltonian system, in general. Later, we will focus on those that are not broken by magnetic field.
\begin{prop}
\label{prop:0}
Take the 6-dimensional space of a single particle, with coordinates and momenta $(x,y,z,p_x,p_y,p_z)$, and let $I$ be the identity operator on this space. The general form of a time reversal operator $\mathcal{T}$, for classical Hamiltonian dynamics, writes: 
\begin{equation}
\label{eq:2.4}
\mathcal{T}(x,y,z,p_x,p_y,p_z)=P(s_1x,s_2y,s_3z,-s_1p_x,-s_2p_y,-s_3p_z)
\end{equation}
where $P$ is a permutation of coordinates and of their conjugate momenta, such that $P^2=I$,
and $s_i$, which equals $1$ or $-1$, takes opposite values in front of coordinates and momenta.
\end{prop}
\begin{proof}
That $P^2$ be the identity and that $s_i$ be $\pm1$ is imposed by the fact that $\mathcal{T}^2=I$,  i.e., that a time reversal transformation must be involutional. That a coordinate and its respective momentum have opposite sign is imposed by the form of the Hamilton equations:
\begin{equation}
\label{eq:2.3}
\begin{dcases}
\frac{\partial H}{\partial p_i}=\dot{x}^i\\
\frac{\partial H}{\partial x^i}=-\dot{p_i}
\end{dcases}
\end{equation}

In fact, assuming that the Hamiltonian itself verifies TRI, an overall minus sign arises when time is reversed. Then, in order to preserve the form of the equations of motion, a minus sign has to distinguish $x^i$ from its
conjugate momentum $p_i$.
\end{proof}
Note that $P$ in Equation \eqref{eq:2.4} is not a permutation of six elements but it acts in a block diagonal way on the coordinates and in the same way on the momenta. For instance, assuming $P$ swaps $x$ and $y$, it~does the same with the corresponding momenta:
\begin{equation}
    \label{permex}
    (x,y,z,p_x,p_y,p_z)\xrightarrow[]{P}(y,x,z,p_y,p_x,p_z)
\end{equation}

This action comes in addition to the compulsory alternation of signs between coordinates and conjugated momenta produced by the $s_i$ factors. 

In order to enumerate how many different time reversal transformations exist, let us represent them in matrix form. As positions and momenta are bound to be distinguished by a minus sign, it~suffices to consider the 3-dimensional space of positions, hence to consider a 3 $\times$ 3 matrix, $\mathcal{M}_d$. The~action of $\cal T$ on the corresponding momenta will be given by $-\mathcal{M}_d$. 

First, suppose $P$ is the identity, so that $\mathcal{M}_d$ takes the diagonal form:
\begin{equation}
\label{eq:2.5}
\mathcal{M}_d=
\begin{pmatrix}
s_1 & 0&0\\
0 &s_2 &0\\
0 &0 & s_3
\end{pmatrix}
\end{equation}

In this case, there are eight possible choices for 
$\cal T$, as shown in Ref.~\cite{rondom}. For example, the usual time reversal operation that preserves the coordinates and reverses the momenta corresponds to $s_1=s_2=s_3=1$.

If, on the other hand, $P\neq I$, the total number of permutations of three elements is the order of the discrete group $S_3$,   i.e., $3!=6$. But the cyclical and the counter-cyclical permutations are not involutions, and only the swap permutations remain:
\begin{equation}
\label{eq:2.6}
\mathcal{M}_{xy}=
\begin{pmatrix}
0 & s_P&0\\
s_P &0 &0\\
0 &0 & s_3
\end{pmatrix}
\end{equation}
\begin{equation}
\label{eq:2.7}
\mathcal{M}_{yz}=
\begin{pmatrix}
s_1 & 0&0\\
0 &0 &s_P\\
0 &s_P & 0
\end{pmatrix}
\end{equation}
\begin{equation}
\label{eq:2.8}
\mathcal{M}_{xz}=
\begin{pmatrix}
0 & 0&s_P\\
0 &s_2 &0\\
s_P &0 & 0
\end{pmatrix}
\end{equation}
where $s_P=\pm 1$ and the subscript on $\cal M$ 
identifies the swap. The non-zero elements in the $2$ $\times$ $2$ permutation blocks must own the same sign to ensure that the transformation squared is the identity. This amounts to 12 transformations: four for each of the matrices \eqref{eq:2.6}, \eqref{eq:2.7} and \eqref{eq:2.8}. Adding these to the previous 8 transformations, we obtain a total of 20 generalized time reversal transformations, that can be used to derive the Onsager reciprocal relations, following e.g., the approach of Ref. \cite{rondom}.

For the invariance of the Hamiltonian, let us directly consider the magnetic field, $\mathbf{B}\ne 0$. First, let the particles of the system be  coupled to $\mathbf{B}$ only, so that  $U(\boldsymbol{X},\boldsymbol{P},\boldsymbol{C})=0$.
As there are $20$ possible transformations for each particle subspace, one can choose a time reversal operation among $20^N$. For~instance, let $\mathcal{M}_1$ and $\mathcal{M}_2$ be two matrices that represent two suitable 
transformations on 6-dimensional subspaces; one may combine them in a single transformation $O$ acting on the entire phase space as:
\begin{equation}
\label{totalphasespace}
    O(\boldsymbol{X},\boldsymbol{P})=(\mathcal{M}_1\boldsymbol{x}_1,-\mathcal{M}_1\boldsymbol{p}_1,...,\mathcal{M}_2\boldsymbol{x}_k,-\mathcal{M}_2\boldsymbol{p}_k,... ,\mathcal{M}_2\boldsymbol{x}_N,-\mathcal{M}_2\boldsymbol{p}_N)
\end{equation}
where a special combination of the two operations has been chosen. By definition, $O$ automatically satisfies the conditions \eqref{timerev1}, and can be used under the Kubo correlation integral. 

To find involutions that act on the entire phase space, not as block diagonal single particle matrices, one may consider non-diagonal time reversal operations, that act on the Hamiltonian \eqref{eq:2.2} exchanging coordinates and momenta of different particles. However, because in general  particles have different masses, $m_i \ne m_j$ for
$i\ne j$, such operations do not qualify as time reversal involutions.
For example, consider the following transformation:\vspace{12pt}
\begin{equation}
    \label{nondiag}
    \resizebox{0.95\hsize}{!}{$(\boldsymbol{x}_1,...,\boldsymbol{x}_j,\boldsymbol{x}_{j+1},...,\boldsymbol{x}_N,\boldsymbol{p}_1,....,\boldsymbol{p}_j,\boldsymbol{p}_{j+1},...,\boldsymbol{p}_N)\xrightarrow{\mathcal{M}_{nd}}(\boldsymbol{x}_1,...,\boldsymbol{x}_{j+1},\boldsymbol{x}_j,...,\boldsymbol{x}_N,-\boldsymbol{p}_1,....,-\boldsymbol{p}_{j+1},-\boldsymbol{p}_j,...,-\boldsymbol{p}_N)$}
\end{equation}
where $\boldsymbol{x}_1=(x_1,y_1,z_1)$. Writing the summation in Equation \eqref{eq:2.2} as:
\begin{equation}
    \label{summation}
    \sum_{i=1}^N\left[\frac{(\boldsymbol{p}_i-q_i\boldsymbol{A}(\boldsymbol{x}_i))^2}{2m_i}\right]=...+\frac{(\boldsymbol{p}_j-q_j\boldsymbol{A}(x_j,y_j,z_j))^2}{2m_j}+\frac{(\boldsymbol{p}_{j+1}-q_{j+1}\boldsymbol{A}(x_{j+1},y_{j+1},z_{j+1}))^2}{2m_{j+1}}+...
\end{equation}
the transformation \eqref{nondiag} yields:
\begin{equation}
    \label{summation2}
    ...+\frac{(\boldsymbol{p}_j+q_j\boldsymbol{A}(x_j,y_j,z_j))^2}{2m_{j+1}}+\frac{(\boldsymbol{p}_{j+1}+q_{j+1}\boldsymbol{A}(x_{j+1},y_{j+1},z_{j+1}))^2}{2m_{j}}+...
\end{equation}

{As the transformation \eqref{nondiag} does not act on the masses, Equation \eqref{summation2} may differ from 
the corresponding term in Equation \eqref{summation}, even in cases in which 
$\boldsymbol{A}(x_j,y_j,z_j)=\boldsymbol{A}(x_{j+1},y_{j+1},z_{j+1})$: the~Hamiltonian is not invariant under the action of $\mathcal{M}_{nd}$. }
Depending on the values of the particles masses, certain swaps may be allowed or not. In the following, we limit our investigation to the case that excludes particles swaps.

 Considering the 20 operations listed above, \eqref{eq:2.5}, \eqref{eq:2.6}, \eqref{eq:2.7} and \eqref{eq:2.8}, let us now relate them to the functional form of the vector potential of Equation \eqref{eq:2.2}. Neglecting for sake of simplicity the particle index $i$, we have:
\begin{equation}
    \label{eq:2.9}
    (\boldsymbol{p}-q\boldsymbol{A})^2=(p_x-q A_1)^2+(p_y-q A_2)^2+(p_z-q A_3)^2
\end{equation}

Under the action of the map \eqref{eq:2.4} with $P=I$, this yields:
\begin{equation}
    \label{eq:2.10}
    (-s_1p_x-q A_1(s_1x,s_2y,s_3z))^2+(-s_2p_y-q A_2(s_1x,s_2y,s_3z))^2+(-s_3p_z-q A_3(s_1x,s_2y,s_3z))^2
\end{equation}
and imposing that the result equals the expression \eqref{eq:2.9}, 
\begin{equation}
    \label{eq:2.11}
    (p_x-q A_1)^2+(p_y-q A_2)^2+(p_z-q A_3)^2=(p_x+q s_1 A_1^T)^2+(p_y+qs_2 A_2^T)^2+(p_z+qs_3 A_3^T)^2
\end{equation}
where the $A^T_k$ is the transformed component $A_k(s_1x,s_2y,s_3z)$, the Hamiltonian verifies TRI. We can thus~write:
\begin{prop}
\label{prop:nec}
The necessary and sufficient algebraic conditions for the validity of Equation \eqref{eq:2.11} are given~by:
\begin{equation}
    \label{eq:2.12}
    A_k^T=-s_kA_k \;\;\;\;\;\;\;\;\;\;\; k=1,2,3
\end{equation}
\end{prop}
\begin{proof}
On the one hand, if \eqref{eq:2.12} holds, substitution immediately yields \eqref{eq:2.11}. Vice versa, starting from the validity of \eqref{eq:2.11}, one notes that the squares of $\boldsymbol p$ and $\boldsymbol A$ are squared norms of vectors in $\mathbb{R}^3$, hence are invariant under rotations, as the generalized time reversal operations are. Consequently, the~following equality holds:
\begin{equation}
    \label{nec:1}
    -p_xA_1-p_yA_2-p_zA_3=p_xs_1A_1^T+p_ys_2A_2^T+p_zs_3A_3^T
\end{equation}

As each $A_k$ only depends on $(x,y,z)$, and the conjugate momenta are independent, one may vary at will the values of $(p_x,p_y,p_z)$ in \eqref{nec:1}. Setting to zero two of them, one gets \eqref{eq:2.12} for the third. Repeating, for the~other pairs, \eqref{eq:2.12} is obtained.
\end{proof}
Actually, TRI in presence of a magnetic field is less demanding than that, because it suffices that~\eqref{eq:2.11} holds up to a gauge transformation. In other words, \eqref{eq:2.9} can be generally replaced by:
\begin{equation}
    \label{eq:2.9bis}
    [\boldsymbol{p}-q(\boldsymbol{A}+\nabla G)]^2=[p_x-q (A_1+\partial_x G)]^2+[p_y-q (A_2+\partial_y G)]^2+[p_z-q (A_3+\partial_z G)]^2
\end{equation}
where $G$ is a suitable scalar function that can be introduced without affecting the dynamics.
\begin{prop}
\label{prop:necbis}
Admitting possible gauge transformations, the  necessary and sufficient algebraic conditions for the time reversal invariance of Equation \eqref{eq:2.11} are expressed by:
\begin{equation}
    \label{eq:2.12bis}
    A_k^T=-s_k(A_k+\partial_i G) \;\;\;\;\;\;\;\;\;\;\; k=1,2,3 \text{  and  } i=x,y,z
\end{equation}
\end{prop}
\begin{proof}
The reasoning used in the proof of Proposition \ref{prop:nec} can be repeated. Introducing $A_i+\partial_i G$ in place of $A_i$, in the left hand side of Equation \eqref{eq:2.11}, we get:
\begin{equation}
    \label{eq:2.11tris}
    \begin{split}
         &(p_x-q (A_1+\partial_x G))^2+(p_y-q (A_2+\partial_y G))^2+(p_z-q (A_3+\partial_z G))^2=\\
         &(p_x+q s_1 A_1^T)^2+(p_y+qs_2 A_2^T)^2+(p_z+qs_3 A_3^T)^2
    \end{split}
\end{equation}

Then, direct substitution shows that \eqref{eq:2.12bis} implies \eqref{eq:2.11tris}. The inverse implication follows from the fact that Equation \eqref{eq:2.11tris} has to hold for any value of the coordinates and the momenta. In particular, considering the case $p_x=p_y=p_z=0$, we have:
\begin{equation}
    \label{norm}
    (\boldsymbol A+\nabla G)^2=[\boldsymbol{A}^T]^2
\end{equation}
and trivially the following:
\begin{equation}
    \label{nec:1bis}
    -p_x(A_1+\partial_x G)-p_y(A_2+\partial_y G)-p_z(A_3+\partial_z G)=p_xs_1A_1^T+p_ys_2A_2^T+p_zs_3A_3^T
\end{equation}

The thesis follows separately considering pairs in which two among $p_x$, $p_y$ and $p_z$ vanish.
\end{proof}

As an example, take a constant magnetic field along the $z$ axis, which corresponds to a vector potential $\boldsymbol{A}(x,y,z)= {A_0(0,x,0)=(0,A_0 x, 0)}$, and choose the Coulomb gauge. Then \eqref{eq:2.12} reduces to $s_1x=-s_2x$ for any value of $x$, that is:
\begin{equation}
    \label{eq:2.13}
    s_1=-s_2 
\end{equation}

In this case, the number of diagonal time reversal operations that preserve TRI is four, Ref. \cite{rondom}. Indeed, every constraint on the values of $s_i$ halves the number of available reversal operations. 
Then,  {applying the transformation \eqref{eq:2.6}  to \eqref{eq:2.9} yields (the same can be repeated for \eqref{eq:2.7} and \eqref{eq:2.8}):} 
\begin{equation}
    \label{eq:2.14}
    (-s_Pp_x-q A_2(s_Py,s_Px,s_3z))^2+(-s_Pp_y-q A_1(s_Py,s_Px,s_3z))^2+(-s_3p_z-q A_3(s_Py,s_Px,s_3z))^2
\end{equation}
and in the same way as Proposition \ref{prop:nec} we derive three necessary and sufficient conditions
\begin{equation}
    \label{eq:2.15}
    \begin{dcases}
    A_1(s_Py,s_Px,s_3z)=-s_PA_2(x,y,z)\\
    A_2(s_Py,s_Px,s_3z)=-s_PA_1(x,y,z)\\
    A_3(s_Py,s_Px,s_3z)=-s_zA_3(x,y,z)
    \end{dcases}
\end{equation}

In the singular case $\boldsymbol{A}(x,y,z)=(0,{A_0x},0)$, \eqref{eq:2.15} reduces to $0=\pm s_Px$, which clearly has no solution for $s_P=\pm1$; on the other hand, one observes
that the same magnetic field corresponds to the vector potential $\boldsymbol{A}(x,y,z)= {A_0}/2(-y,x,0)$, that instead leads to 
\begin{equation}
    \label{eq:2.16}
    \begin{dcases}
    -s_Px=-s_Px\\
    s_Py=s_Py
    \end{dcases}
\end{equation}
which has solution. In other words, the four transformations in the form \eqref{eq:2.6} continue to hold.
The~point is that one can use the gauge freedom to replace \eqref{eq:2.15}, and write:
\begin{equation}
    \label{eq:2.15bis}
    \begin{dcases}
    A_1(s_Py,s_Px,s_3z)=-s_P(A_2(x,y,z)+\partial_y G)\\
    A_2(s_Py,s_Px,s_3z)=-s_P(A_1(x,y,z)+\partial_x G)\\
    A_3(s_Py,s_Px,s_3z)=-s_z(A_3(x,y,z)+\partial_z G)
    \end{dcases}
\end{equation}

In the next section,
we discuss in detail the role of the gauge. 

\subsection{Gauge}
\label{sec:Gauge}
By definition, the gauge choice has no physical consequences. In our case, the dynamics does not change if the vector potential $\boldsymbol{A}$ is replaced by $\boldsymbol{A}+\nabla G$, with $G:\mathbb{R}^3\xrightarrow{}\mathbb{R}$ a scalar function. As~commonly done in this kind of magnetostatic problems, we choose the Coulomb gauge:
\begin{equation}
    \label{eq:2.17}
    \nabla\cdot\boldsymbol{A}=0
\end{equation}

The consequence of this on the physical field $\boldsymbol{B}$, hence on the conditions for TRI, can be illustrated starting from the diagonal transformations and recasting \eqref{eq:2.12bis} in the following fashion:
\begin{equation}
    \label{eq:2.18}
    (s_1A_1^T,s_2A_2^T,s_3A_3^T)=-(A_1+\partial_x G,A_2+\partial_y G,A_3+\partial_z G) = -(\boldsymbol{A}+\nabla G)
\end{equation}
where we used the fact that \eqref{eq:2.4} has to be an involution.

One can view Equation \eqref{eq:2.4} (with $P=I$ in the diagonal case) as a transformation on the  vector field $V(\mathbb{R}^3)$ of which $\boldsymbol{A}$ is an element, that transforms as a vector and not as a pseudo-vector. Hence, the~necessary conditions \eqref{eq:2.12} imply that $\boldsymbol{A}$ transformed as a vector field in $\mathbb{R}^3$ under a diagonal operation $\mathcal{M}':V(\mathbb{R}^3)\xrightarrow{}V(\mathbb{R}^3)$ has to equal $-\boldsymbol{A}$ up to a gauge transformation, and $\boldsymbol{B}$ is then mapped~to~$-\boldsymbol{B}$.

The same applies to the non diagonal transformations: we rewrite \eqref{eq:2.15bis} as 
\begin{equation}
    \label{eq:2.19}
    \begin{pmatrix}
    A_1'\\A_2'\\A_3'
    \end{pmatrix}
    =-\mathcal{M}_{xy}
    \begin{pmatrix}
    A_1+\partial_x G\\A_2+\partial_y G\\A_3+\partial_z G
    \end{pmatrix}
\end{equation}
where $A_k'=A_k(s_Py,s_Px,s_3z)$. As the inverse of the matrix $\mathcal{M}_{xy}$ equals the matrix itself, multiplying Equation \eqref{eq:2.19} side by side by $\mathcal{M}_{xy}$ the consequence is again to transform $\boldsymbol{A}$ into $-\boldsymbol{A}$ up to a gauge transformation. The same obviously holds for $\mathcal{M}_{xz}$ and $\mathcal{M}_{yz}$.

The gauge freedom can be accounted for by introducing the equivalence classes $\left[\boldsymbol{A}\right]$ of the vector potentials that lead to the same magnetic fields,  i.e., whose elements differ by the gradient of an at least twice differentiable scalar function $G(x,y,z)$. We denote by $[\boldsymbol A]_R$ an element of the class $[\boldsymbol A]$, that~corresponds to a particular choice of $G$. We can now state the following:
\begin{prop}
\label{prop:2.1}
A generalized time reversal operation $\mathcal{M}$ of form \eqref{eq:2.4}, that acts on all particles 6-dimensional subspaces, preserves TRI in the presence of a magnetic vector potential $\boldsymbol{A}$ if and only if the associated transformation defined on the 3-dimensional vector field space, $\mathcal{M}':V(\mathbb{R}^3)\xrightarrow{}V(\mathbb{R}^3)$, obeys: 
\begin{equation}
\label{eq:2.20}
    \mathcal{M}'\boldsymbol{A}=\mathcal{M}_M(A_1(\mathcal{M}_M\boldsymbol{x}),A_2(\mathcal{M}_M\boldsymbol{x}),A_3(\mathcal{M}_M\boldsymbol{x}))=[-\boldsymbol{A}]_R
\end{equation} 
with $\mathcal{M}_M$ one 3-dimensional specific matrix representation verifying $\mathcal{M}_M^2=I$. 
\end{prop}
\noindent
When this is verified, the Hamiltonian is preserved up to a gauge transformation and the corresponding equations of motion are in turn verified.
\begin{proof}
The direct implication directly comes from Equations \eqref{eq:2.18} and   \eqref{eq:2.19}, where the invariance of the equations of motion leads to the condition \eqref{eq:2.20}. Vice versa, assuming  there is an involution $\mathcal{M}'$ verifying Equation \eqref{eq:2.20}, with 3-dimensional matrix representation $\mathcal{M}_M$, one can introduce
the transformation $\mathcal{M}\equiv(\mathcal{M}_M\boldsymbol{x},-\mathcal{M}_M\boldsymbol{p})$, which preserves the structure of the Hamilton equations under time reversal, because it alternates signs. Furthermore, the Hamiltonian is unchanged under the application of $\mathcal{M}$ to every particle space, since  $\mathcal{M}(\boldsymbol{p}-q\boldsymbol{A}(\boldsymbol{x}))^2=(-\mathcal{M}_M\boldsymbol{p}-q\boldsymbol{A}(\mathcal{M}_M\boldsymbol{x}))^2$ by definition. Using \eqref{eq:2.20} and $\mathcal{M}_M^2=I$ we obtain $\mathcal{M}(\boldsymbol{p}-q\boldsymbol{A}(\boldsymbol{x}))^2=(\boldsymbol{p}-q[\boldsymbol{A}(\boldsymbol{x})]_R)^2$.
\end{proof}

\begin{rem}
\label{rem:2.1}
{\it 
Applying $\mathcal{M}$ as a variable change in the integral \eqref{correl:2} deeply differs from  inverting $\boldsymbol{B}$. The~coordinates swap operated by $\mathcal{M}$ may amount to a mere rearrangement of the order in which the contributions to the integral coming from the different regions of the phase space are summed up, that does not affect the total. That depends on the functions that are integrated. 
For instance, given an average electric current from left to right, corresponding to a forward trajectory of particles, its time reverse may exist even if the particles do not trace backward the configurations of the forward trajectory; a reversed average of momenta suffices.}
\end{rem}

\begin{rem}
\label{rem:2.1bis}
{\it
Remark \ref{rem:2.1} rests on the hypothesis that all coordinate transformations of interest map the domain of integration on itself. Depending on the geometry of interest, a coordinate change may kick some particle out of the volume occupied by the system under 
investigation. As long as one remains within the realm of infinite homogeneous systems, or far from possible boundaries, as common in response theory, this is not an issue. In~general, one has to consider case by case whether the phase space is invariant under the chosen time reversal mapping. If the dynamics is not translation invariant, making all time reversal symmetries fail, in principle one obtains a method to experimentally find a violation of Onsager reciprocal relations.}
\end{rem}
To test the condition of Proposition \ref{prop:2.1}, it suffices to check that the curl of $\boldsymbol{A}$ and of $\mathcal{M}'\boldsymbol{A}$ corresponds {to $\boldsymbol{B}$ and $-\boldsymbol{B}$, respectively}. For example, take a constant magnetic field with gauge choices $\boldsymbol{A}_1(x,y,z)=(0,{A_0 x},0)$ and $\boldsymbol{A}_2(x,y,z)=(-{A_0 y},0,0)$, which are elements of the same class $\left[\boldsymbol{A}\right]$.
Applying the transformation of Equation \eqref{eq:2.6} with $s_P=1$ and $s_3=1$, one obtains $\boldsymbol{A}'_2(x,y,z)=(0,-{A_0 x},0)$ that does not equal $-\boldsymbol{A}_2(x,y,z)$, but equals $-\boldsymbol{A}_1(x,y,z)$, showing that it nevertheless belongs to the class $\left[-\boldsymbol{A}\right]$. Thus, the transformation of Equation \eqref{eq:2.6} satisfies the  necessary condition \eqref{eq:2.20} for TRI.

\subsection{Magnetic field}
\label{sec:Magnetic Field}
Proposition \ref{prop:2.1} can be formulated in an equivalent form that does not involve gauge freedom:
\begin{prop}
\label{prop:2.2}
A generalized time reversal operation $\mathcal{M}$ of form \eqref{eq:2.4}, that acts on all particles 6-dimensional subspaces, preserves TRI in the presence of a magnetic field $\boldsymbol{B}$ if and only if the associated transformation defined on the 3-dimensional vector field space, $\mathcal{M}':V(\mathbb{R}^3)\xrightarrow{}V(\mathbb{R}^3)$, obeys:
\begin{equation}
\label{eq:2.23}
    \mathcal{M}'\boldsymbol{B}=det(\mathcal{M}_M)\mathcal{M}_M(B_1(\mathcal{M}_M\boldsymbol{x}),B_2(\mathcal{M}_M\boldsymbol{x}),B_3(\mathcal{M}_M\boldsymbol{x}))=-\boldsymbol{B}
\end{equation} 
with $\mathcal{M}_M$ the 3-dimensional specific matrix representation verifying $\mathcal{M}_M^2=I$. 
\end{prop}
\begin{proof}
The derivation is trivial because \eqref{eq:2.20} and \eqref{eq:2.23} are equivalent statements by definition of  {a magnetic field} as curl of  {vector potential}, which transforms as a pseudo-vector in 3D space.
\end{proof}

Again, TRI preserves the Hamiltonian, up to a gauge choice, as well as the corresponding equations of motion.
This perspective is particularly useful in classical mechanics, in which only the magnetic field matters, because the equations of motion are the fundamental element of the theory. 

Now, given a magnetic field $\boldsymbol{B}(\boldsymbol{x})$, the necessary conditions for a transformation 
to preserve TRI are obtained from Equations \eqref{eq:2.5}, \eqref{eq:2.6}, \eqref{eq:2.7} or \eqref{eq:2.8}. To do that for the 20 transformations we have got, let us express $\boldsymbol{B}$ in the basis $\boldsymbol{\hat i},
\boldsymbol{\hat j},\boldsymbol{\hat k}$ 
of the 3-dimensional space as:
\begin{equation}
    \label{eq:2.24}
    \boldsymbol{B}=B_1(\boldsymbol{x})\boldsymbol{\hat i}+B_2(\boldsymbol{x})\boldsymbol{\hat j}+B_3(\boldsymbol{x})\boldsymbol{\hat k}
\end{equation}
and take the diagonal transformations with matrix representation \eqref{eq:2.5}. Following the rule \eqref{eq:2.23}, $\boldsymbol{B}$~transforms as: 
\begin{equation}
    \label{eq:2.25}
    \boldsymbol{B}'=s_1s_2s_3[s_1B_1(s_1x,s_2y,s_3z)\boldsymbol{\hat i}+s_2B_2(s_1x,s_2y,s_3z)\boldsymbol{\hat j}+s_3B_3(s_1x,s_2y,s_3z)\boldsymbol{\hat k}]
\end{equation}

Then, the necessary matching conditions between the magnetic field components and the transformation  follow from the second equality of \eqref{eq:2.23}, and write:
\begin{equation}
\label{eq:2.26}
\begin{cases}
B_1(x,y,z)=-s_2s_3B_1(s_1x,s_2y,s_3z)\\
B_2(x,y,z)=-s_1s_3B_2(s_1x,s_2y,s_3z) \\
B_3(x,y,z)=-s_1s_2B_3(s_1x,s_2y,s_3z)
\end{cases}
\end{equation}

Therefore, given the magnetic field, one can verify by inspection which of the eight diagonal transformations yield TRI. The same reasoning can be repeated for the non diagonal transformations, with representations \eqref{eq:2.6}, \eqref{eq:2.7} or \eqref{eq:2.8}, whose application to \eqref{eq:2.24} implies:
\begin{equation}
\label{condit1}
\boldsymbol{B}'_{xy}=-s_3[s_PB_2(s_Py,s_Px,s_3z)\boldsymbol{\hat i}+s_PB_1(s_Py,s_Px,s_3z)\boldsymbol{\hat j}+s_3B_3(s_Py,s_Px,s_3z)\boldsymbol{\hat k}]
\end{equation}
\begin{equation}
\label{condit2}
\boldsymbol{B}'_{yz}=-s_1[s_1B_1(s_1x,s_Pz,s_Py)\boldsymbol{\hat i}+s_PB_3(s_1x,s_Pz,s_Py)\boldsymbol{\hat j}+s_PB_2(s_1x,s_Pz,s_Py)\boldsymbol{\hat k}]
\end{equation}
\begin{equation}
\label{condit3}
\boldsymbol{B}'_{xz}=-s_2[s_PB_3(s_Pz,s_2y,s_Px)\boldsymbol{\hat i}+s_2B_2(s_Pz,s_2y,s_Px)\boldsymbol{\hat j}+s_PB_1(s_Pz,s_2y,s_Px)\boldsymbol{\hat k}]\end{equation}
where the subscripts identify the transformation. This derives from the fact that the determinant of the matrices \eqref{eq:2.6}, \eqref{eq:2.7} and \eqref{eq:2.8} equals the opposite of the diagonal element: $-s_P^2s_i=-s_i$. Then, the~necessary matching conditions for the 12 non-diagonal reversal operators write: 
\begin{equation}
\label{eq:2.27}
\begin{cases}
B_1(x,y,z)=s_3s_PB_2(s_Py,s_Px,s_3z)\\
B_2(x,y,z)=s_3s_PB_1(s_Py,s_Px,s_3z) \\
B_3(x,y,z)=B_3(s_Py,s_Px,s_3z)
\end{cases}
\end{equation}
\begin{equation}
\label{eq:2.28}
\begin{cases}
B_1(x,y,z)=B_1(s_1x,s_Pz,s_Py)\\
B_2(x,y,z)=s_1s_PB_3(s_1x,s_Pz,s_Py) \\
B_3(x,y,z)=s_1s_PB_2(s_1x,s_Pz,s_Py)
\end{cases}
\end{equation}
\begin{equation}
\label{eq:2.29}
\begin{cases}
B_1(x,y,z)=s_2s_PB_3(s_Pz,s_2y,s_Px)\\
B_2(x,y,z)=B_2(s_Pz,s_2y,s_Px) \\
B_3(x,y,z)=s_2s_PB_1(s_Pz,s_2y,s_Px)
\end{cases}
\end{equation}

This concludes the case of systems with $U(\boldsymbol{X},\boldsymbol{P},\boldsymbol{C})=0$ in the Hamiltonian. For $U(\boldsymbol{X},\boldsymbol{P},\boldsymbol{C}) \ne 0$, TRI requires also the following:
\begin{equation}
\label{eq:2.41}
   \mathcal{M}U(\boldsymbol{X},\boldsymbol{P},\boldsymbol{C})=U(\mathcal{M}_C\boldsymbol{X},-\mathcal{M}_C\boldsymbol{P},\boldsymbol{C})=U(\boldsymbol{X},\boldsymbol{P},\boldsymbol{C})
\end{equation}
where $\mathcal{M}$ is a time reversal transformation on the phase space, obtained by applying a given $\mathcal{M}_C$ to the coordinates, and alternating signs with the momenta. 
Let us begin introducing a force $\boldsymbol E$ deriving from a scalar potential $\Phi$ that depends only on coordinates, so that $-\nabla\Phi=\boldsymbol F$, and the Hamiltonian~reads:
\begin{equation}
    \label{eq:2.42}
    H=\sum_i^N\left[{\frac {[\boldsymbol {p}_i -q_i\boldsymbol {A}(x_i,y_i,z_i) )]^{2}}{2m_i}} + \Phi(x_i,y_i,z_i)\right]
\end{equation}

Given a transformation $\mathcal{M}$ that satisfies the conditions of Proposition \ref{prop:2.2}, the Hamiltonian \eqref{eq:2.42} results invariant under the application of $\mathcal{M}$ if:
\begin{equation}
    \label{eq:2.43}
     \mathcal{M}\Phi(\boldsymbol{X})=\Phi(\mathcal{M}_C\boldsymbol{X})=\Phi(\boldsymbol{X})
\end{equation}
and $\mathcal{M}_C$ is used as in Equation \eqref{eq:2.41} ({\em n.b.} this includes the notable case of the coupling with an electric field). In the following Section, we investigate notable examples of force potentials.

\subsection{Force Potentials}
\label{sec:Force potentials}
In this Section we consider physically relevant inter-particle potentials. Without loss of generality, we take a constant magnetic fields along the $z$ axis,  i.e., $\boldsymbol{B}=(0,0,1)$, which breaks four of the eight diagonal time reversal symmetries.
In turn, the conditions \eqref{eq:2.27}, \eqref{eq:2.28} and \eqref{eq:2.29} imply that only the four non diagonal operations \eqref{eq:2.6} yield TRI, producing a total of eight time reversal symmetries.

\begin{exmp}
\normalfont
Take a central potential, e.g., the Coulomb potential between charged particles:
\begin{equation}
    \label{eq:2.44}
    U(\boldsymbol{X},\boldsymbol{P},\boldsymbol{C})=\sum_{i<j}^N f_{ij}(\boldsymbol{C})u(r_{ij}) \, ; \quad
r_{ij}=\sqrt{(x_i-x_j)^2+(y_i-y_j)^2+(z_i-z_j)^2} \, ,
\end{equation}
$r_{ij}$ being the distance between particle $i$ and particle $j$, $\boldsymbol{C}$ a vector of parameters, and $f_{ij}$ a function of such parameters. This potential satisfies the condition \eqref{eq:2.43} because each of the 20 available transformations $\mathcal{M}_C$ is an element of the orthogonal group $O(3)$. In particular, one may take block diagonal operators with $3\times3$ blocks given by \eqref{eq:2.5}, \eqref{eq:2.6}, \eqref{eq:2.7} or \eqref{eq:2.8}. 
As a consequence, $r_{ij}$ is left unchanged by the action of $\mathcal{M}$ on the phase space. Moreover, $\mathcal{M}_C$ does not act on the space of the parameters $\boldsymbol{C}$, leaving each $f_{ij}$ invariant. 

While very simple, the potentials of this form are most common and useful; in particular, interactions between structureless objects are commonly modelled by central forces, such as those derived from Lennard--Jones, Morse, Coulomb, gravitational and Yukawa potentials.\end{exmp}
\begin{exmp}
\normalfont
The Coulomb ring-shaped (or Hartmann) potential treated in Ref. \cite{ring}
\begin{equation}
    \label{eq:2.46}
    U(x_i,y_i,z_i)=-\frac{Z}{\sqrt{x_i^2+y_i^2+z_i^2}}+\frac{1}{2}Q\frac{1}{x_i^2+y_i^2}\;\;\;\;\; Q>0 \,, \,\, Z>0
\end{equation}
is used in quantum mechanics, and can be used to model a force field that is not purely central, thanks to its second addend, that depends on the square distance from $z$ axis. Here, the term $x_i^2+y_i^2$ is  invariant under the action of the 8 possible diagonal transformations; in particular, we have:
\begin{equation}
    \label{eq:2.47}
    (s_1x_i)^2+(s_2y_i)^2=x_i^2+y_i^2
\end{equation}

In addition, for the non-diagonal transformations of the form \eqref{eq:2.6}, we have:
\begin{equation}
    \label{eq:2.48}
    (s_Py_i)^2+(s_Px_i)^2=x_i^2+y_i^2
\end{equation}

In conclusion, this kind of potential does not add restrictions to TRI, other than those imposed by the  magnetic~field.
\end{exmp}
\begin{exmp}
\normalfont
A different kind of potentials, used, e.g., in molecular dynamics, depends on momenta. For~instance, in~Ref.~\cite{ferm}, classical Fermion-like particles are simulated with the following potential:
\begin{equation}
\label{eq:2.45}
    U(\boldsymbol{p}_i)=\frac{E_p}{1+e^{b_p (|\boldsymbol{p}_i|^2-1)}}
\end{equation} 
where $E_p$ and $b_p$ are dimensional constants, while $\boldsymbol{p}_i=(p_i^x,p_i^y,p_i^z)$. In this case, the particles are decoupled, but they are subject to an external momentum dependent force. TRI, hence its consequences such as Onsager reciprocal relations, may hold even in a system like this, if the functional form of the magnetic field allows, because $|\boldsymbol{p}_i|$ is invariant under rotations.
\end{exmp}

\begin{exmp}
\normalfont
The Polarisable Ion Model (PIM) potential, is particularly interesting in molecular dynamics studies, to take into account certain intermolecular interactions cf.\  Refs. \cite{ishii2015transport,tesson2016classical}. In the case 
of an $N$ particles system, it is expressed by:
\begin{equation}
\label{eq:2.49}
U=U_{charge}+U_{dispersion}+U_{repulsion}+U_{polarization}
\end{equation}
where
\begin{equation}
\label{eq:2.50}
U_{charge}=\sum_i \sum_{j>i}\frac{q_i q_j}{r_{ij}}
\end{equation}
is the Coulomb electric potential,
\begin{equation}
\label{eq:2.51}
U_{dispersion}=-\sum_i\sum_{j>i}\left(\frac{C^{ij}_6}{(r_{ij})^6}f_6^{ij}(r_{ij})+\frac{C^{ij}_8}{(r_{ij})^8}f_8^{ij}(r_{ij})\right)
\end{equation}
is due to dipole-dipole and dipole-quadrupole dispersion,
\begin{equation}
\label{eq:2.52}
U_{repulsion}=\sum_i\sum_{j>i}B_{ij}e^{-\alpha_{ij}r_{ij}}
\end{equation}
is a short-range repulsion term, and
\begin{equation}
\label{eq:2.53}
\begin{split}
U_{polarization}=&\sum_i\sum_{j>i}\left(\frac{q_i \boldsymbol{r_{ij}\cdot\mu_j}}{(r_{ij})^3}f^{ij}_4(r_{ij})-\frac{q_j \boldsymbol{r_{ij}\cdot\mu_i}}{(r_{ij})^3}f^{ji}_4(r_{ij}) \right)+\\
&\sum_i\sum_{j>i}\left(\frac{\boldsymbol{\mu_i\cdot\mu_j}}{(r_{ij})^3}-\frac{3(\boldsymbol{r_{ij}\cdot\mu_i})(\boldsymbol{r_{ij}\cdot\mu_j})}{(r_{ij})^5} \right)+\sum_i\frac{|\boldsymbol{\mu_i}|^2}{2\alpha_i}
\end{split}
\end{equation}
is the polarization interaction term, with $\boldsymbol{\mu_i}$ the induced dipole moment of the molecule $i$. 
While the parts in Equations \eqref{eq:2.50}, \eqref{eq:2.51} and \eqref{eq:2.52} are like the potential \eqref{eq:2.44}, and are invariant under any time reversal operation, the term in Equation \eqref{eq:2.53} is hard to control, since it is defined recursively: for~any particle $i$, $\boldsymbol{\mu_i}$ in principle depends on the coordinates and on the dipole momenta of all the other particles. Explicitly expressing this dependence is problematic, and the verification of Equation~\eqref{eq:2.41} so far remains out of reach. In fact, this potential is only analyzed through approximations and~numerically.
\end{exmp}
\section{Conclusions}
\label{sec:Conclusions}

In this article, we have generalized the results of Refs. \cite{rondoe,rondom,bonella2017time,casa},  increasing the number of time reversal symmetries that concern mechanical systems in general, and systems in magnetic field, in~particular. We~focused on block diagonal transformations, composed by operations acting on the 6-dimensional subspace of each particle, and we have introduced suitable equivalence classes to account for the corresponding gauge invariance. We then obtained sufficient conditions for TRI to hold in presence of a magnetic field, which 
imply, for instance, Onsager reciprocal relations. Substantially enlarging the range of applicability of TRI, we contribute to understand why violations of such relations to date are not reported, despite the presence of magnetic fields.

The next step will be to investigate the necessary conditions for the validity of Onsager reciprocal relations. Indeed, as Ref. \cite{casa} states, the discovery of a violation of Onsager reciprocal relations may lead to the never observed situation of non-dissipative currents. This may be a dynamically indirect reason why Onsager reciprocal relations cannot be broken, at least in classical systems where the evidence of superconductivity was never found. 

In the final part of this paper, we have illustrated  the application of our results to notable potentials. Such a few examples do not exhaust the set of possible situations in which TRI holds or is violated, both theoretically and experimentally. However, it
covers typical situations and constitutes a guide for further investigations of the Onsager reciprocal relations.

 {As pointed out by one of the anonymous referees,
electromagnetism is inherently relativistic, hence in future works we may investigate the extension of our present results to the relativistic case. As a matter of fact, regarding the time reversal operations on the single particle subspace, thus of any set of non-interacting particles,
a formal extension of our involutions 
is immediate, although not necessarily conceptually satisfactory,
given the role of time in Minkowski space.
Moreover,
Statistical Mechanics relations, such as those considered in this paper, require interacting particles. This
makes
the subject most intriguing and challenging \cite{dunkel2009relativistic,mi2011introduction,AliR}.}

\vspace{6pt} 



\section*{Acknowledgement}
The authors are grateful to Alessandro Coretti, for carefully reading the original draft of this paper and for most insightful remarks.\\
LR has been partially supported by Ministero dell’Istruzione dell’Università e della Ricerca (MIUR) grant Dipartimenti di Eccellenza 2018–2022.





\begin{thebibliography}{999}
\providecommand{\natexlab}[1]{#1}

\bibitem{onsager1931reciprocal}
Onsager, L.
\newblock Reciprocal relations in irreversible processes. I.
\newblock {\em Phys. Rev.} {\bf 1931}, {\em 37},~405.

\bibitem{onsager1931reciprocal2}
Onsager, L.
\newblock Reciprocal relations in irreversible processes. II.
\newblock {\em Phys. Rev.} {\bf 1931}, {\em 38},~2265.

\bibitem{RevModPhys.17.343}
Casimir, H.B.G.
\newblock On Onsager's Principle of Microscopic Reversibility.
\newblock {\em Rev. Mod. Phys.} {\bf 1945}, {\em 17},~343--350.

\bibitem{kubo1957statistical}
Kubo, R.
\newblock Statistical-mechanical theory of irreversible processes. I. General
  theory and simple applications to magnetic and conduction problems.
\newblock {\em J. Phys. Soc. Jpn.} {\bf 1957}, {\em
  12},~570--586.

\bibitem{kubo1959some}
Kubo, R.
\newblock \emph{Some Aspects of the Statistical-Mechanical Theory of Irreversible
  Processes,} In:  \emph{Lecture in Theoretical Physics, } W. Brittin, Interscience:  New York, NY, USA,
   1959.


\bibitem{kubo1966fluctuation}
Kubo, R.
\newblock The fluctuation-dissipation theorem.
\newblock {\em Rep. Prog. phys.} {\bf 1966}, {\em 29},~255.

\bibitem{lax}
Lax, M.J.
\newblock {\em Symmetry Principles in Solid State and Molecular Physics}; John
  Wiley and Sons:
  New York, NY, USA, 1974.


\bibitem{rondoe}
Bonella, S.; Ciccotti, G.; Rondoni, L.
\newblock Time reversal symmetry in time-dependent correlation functions for
  systems in a constant magnetic field.
\newblock {\em EPL (Europhys. Lett.)} {\bf 2014}, {\em 108},~60004.

\bibitem{rondom}
Coretti, A.; Bonella, S.; Rondoni, L.; Ciccotti, G.
\newblock Time reversal and symmetries of time correlation functions.
\newblock {\em Mol. Phys.} {\bf 2018}, {\em 116},~3097--3103.

\bibitem{bonella2017time}
Bonella, S.; Coretti, A.; Rondoni, L.; Ciccotti, G.
\newblock Time-reversal symmetry for systems in a constant external magnetic
  field.
\newblock {\em Phys. Rev. E} {\bf 2017}, {\em 96},~012160.

\bibitem{casa}
Luo, R.; Benenti, G.; Casati, G.; Wang, J.
\newblock Onsager reciprocal relations with broken time-reversal symmetry.
\newblock {\em Phys. Rev. Res.} {\bf 2020}, {\em 2},~022009.

\bibitem{rondoquantu}
De~Gregorio, P.; Bonella, S.; Rondoni, L.
\newblock Quantum Correlations under Time Reversal and Incomplete Parity
  Transformations in the Presence of a Constant Magnetic Field.
\newblock {\em Symmetry} {\bf 2017}, {\em 9},~120.

\bibitem{jacquod2012onsager}
Jacquod, P.; Whitney, R.S.; Meair, J.; B{\"u}ttiker, M.
\newblock Onsager relations in coupled electric, thermoelectric, and~spin
  transport: The tenfold way.
\newblock {\em Phys. Rev. B} {\bf 2012}, {\em 86},~155118.

\bibitem{zirnbauer2010symmetry}
Zirnbauer, M.R.
\newblock Symmetry classes.
\newblock {\em arXiv } {\bf 2010}, arXiv:1001.0722.

\bibitem{Bettolo}
Marini Bettolo~Marconi, U.; Puglisi, A.; Rondoni, L.; Vulpiani, A.
\newblock Fluctuation–dissipation: Response theory in statistical physics.
\newblock {\em Phys. Rep.} {\bf 2008}, {\em 461},~111--195.

\bibitem{ring}
Ya{\c{s}}uk, F.; Berkdemir, C.; Berkdemir, A.
\newblock Exact solutions of the Schr{\"o}dinger equation with non-central
  potential by the Nikiforov--Uvarov method.
\newblock {\em J. Phys. A Math. Gen. } {\bf 2005}, {\em
  38},~6579.

\bibitem{ferm}
Cordero, P.; Hern{\'a}ndez, E.
\newblock Momentum-dependent potentials: Towards the molecular dynamics of
  fermionlike classical particles.
\newblock {\em Phys. Rev. E} {\bf 1995}, {\em 51},~2573.

\bibitem{ishii2015transport}
Ishii, Y.; Kasai, S.; Salanne, M.; Ohtori, N.
\newblock Transport coefficients and the Stokes--Einstein relation in molten
  alkali halides with polarisable ion model.
\newblock {\em Mol. Phys.} {\bf 2015}, {\em 113},~2442--2450.

\bibitem{tesson2016classical}
Tesson, S.; Salanne, M.; Rotenberg, B.; Tazi, S.; Marry, V.
\newblock Classical polarizable force field for clays: Pyrophyllite and talc.
\newblock {\em   J. Phys. Chem. C} {\bf 2016}, {\em
  120},~3749--3758.

\bibitem{dunkel2009relativistic}
Dunkel, J.; H{\"a}nggi, P.
\newblock Relativistic brownian motion.
\newblock {\em Phys. Rep.} {\bf 2009}, {\em 471},~1--73.

\bibitem{mi2011introduction}
Mi~Hakim, R.
\newblock {\em Introduction to Relativistic Statistical Mechanics: Classical
  and Quantum}; World Scientific:
  Singapore,
 2011.

\bibitem{AliR}
Aliano, A.; Rondoni, L.; Morriss, G.
\newblock Maxwell-J{\"u}ttner distributions in relativistic molecular dynamics.
\newblock {\em  Eur.~Phys. J. B-Condens. Matter Complex
  Syst.} {\bf 2006}, {\em 50},~361--365.

\end{thebibliography}

\end{document}